\documentclass[runningheads]{llncs}
\usepackage{graphicx}
\usepackage{epsfig}
\usepackage{graphics}
\usepackage{array}
\usepackage{amsmath}
\usepackage{algorithm}
\usepackage{algorithmic}
\usepackage{fullpage}
\title{Domain Specific Hierarchical Huffman Encoding}
\author{K.Ilambharathi, G.S.N.V.Venkata Manik, N.Sadagopan,B.Sivaselvan}
\institute{Department of Computer Science and Engineering,\\ Indian Institute of Information Technology, Design and Manufacturing, Kancheepuram, Chennai, India. \\
\email{\{coe09b007,coe09b008,sadagopan,sivaselvanb\}@iiitdm.ac.in}}
\begin{document}
\maketitle
\begin{abstract}
In this paper, we revisit the classical data compression problem for domain specific texts.  It is well-known that classical Huffman algorithm is optimal with respect to prefix encoding and the compression is done at character level.  Since many data transfer are domain specific, for example, downloading of lecture notes, web-blogs, etc., it is natural to think of data compression in larger dimensions (i.e. word level rather than character level).  Our framework employs a two-level compression scheme in which the first level identifies frequent patterns in the text using classical frequent pattern algorithms.  The identified patterns are replaced with special strings and to acheive a better compression ratio the length of a special string is ensured to be shorter than the length of the corresponding pattern.
After this transformation, on the resultant text, we employ classical Huffman data compression algorithm.  In short, in the first level compression is done at word level and in the second level it is at character level.  Interestingly, this two level compression technique for domain specific text outperforms classical Huffman technique.  To support our claim, we have presented both theoretical and simulation results for domain specific texts.
\end{abstract}
\section{Introduction}
Data transfer between a pair of nodes is a fundamental problem in computer networks.  Data compression is a technique that speeds up the data transfer by compressing the data at the sender and the original data is recovered at the receiver by employing a decompression technique.   Data compression (decompression) is a classical problem in computer science and it has attracted many researchers in the past and the popular one is due to Huffman.  Huffman data compression technique does character-level compression and does not assume anything about the underlying domain.  Huffman's approach is the following: assign a shorter code for a character which occurs most in the text to be compressed.  Interestingly, this approach is optimal with respect to prefix encoding.   With the discovery of Data mining and in particular, the data compression perspective of data mining looks at the text from a larger dimension and focuses on identifying patterns (words) that occur frequently in the text.   This line of research was initiated in \cite{corpus,bang}. In both approaches there is little assumption about the input text  and the patterns to be searched is precisely from the standard dictionary words.  However, many data transfer operation is domain specific, for example, downloading of lecture notes, web-blogs, etc.  Moroever, we noticed that the data available in the Web Servers (Academic Servers) are tagged or classified according to the domain from which the text is derived, for example, the blog websites or news aggregators running on the web servers display plots (news) tagged with the domain of the text.  A technical blog on programming (computer science) containing posts related to computer science can be better compressed and sent to the readers.  These observations motivate us look at data compression perspective for domain specific text.  Moreover, the existing approaches identify patterns using dictionary as the reference which are not efficient enough for domain specific text as most of the domain specific keywords do not appear in the dictionary.  So, this calls for a different data compression approach for domain specific text.  Therefore, the combinatorial problem at hand is to compress a domain specific text based on the frequency of patterns generated from the text and the objective is to maximize the compression ratio by minimizing the size of file to be transferred between sender and the receiver.  \\
{\bf Related Work:} The study of data compression was initiated by Huffman\cite{aho}.  His technique focuses on character-level compression using the frequency count of the characters in the text.  Huffman's result is very well-known in the literature and it is in use even today.  In \cite{amar}, the possibility of data compression by replacing certain characters in the words by a special character $'*'$ and retaining a few characters so that the word is still retrievable unambiguously were discussed.   It is important to note that any compression technique must be loss less and towards this end Ian et al. proposed a language model for representing the data efficiently through the identification of new tokens and tokens in context of the text under consideration.  A slightly different approach inspired by Pitman Shorthand was adopted in \cite{pitman}.  The approach is to encode a group of successive 2-3 text characters into a single code.   In \cite{pitman}, it is also shown that further application of Huffman coding on the codes generated is possible and is expected to result in a greater compression.   As far as pattern learning and discovery algorithms are concerned, apart from classical frequent pattern mining algorithms \cite{datamining,dm}, the use of Genetic algorithms to arrive at rules or hypotheses for pattern learning in the text compression is presented in \cite{genetic}.  Searching of fixed length patterns in the text can be done very efficiently using the 'TARA' algorithm proposed in \cite{tara}.  Two level dictionary based text compression scheme is proposed in \cite{bang}.  It involves the transformation of the original text with a dictionary of fixed frequent English words.  The disadvantage is the need for a huge dictionary to be present on both the compressor and de-compressor.\\
{\bf Our work:} In this paper, we present a data compression algorithm for a domain specific text.  As mentioned before, many data transfers are domain specific and it is natural to think of domain specific data compression algorithms.  We propose a framework for compression and it works for any domain in general.  For our discussion purposes, we work with {\em computer science} domain.  Our framework employs a two-level compression scheme in which the first level identifies frequent patterns in the text using classical frequent pattern algorithms.  The identified patterns are replaced with special strings and to ensure a better compression ratio the length of a special string is shorter than the length of the corresponding pattern.   After this transformation, on the resultant text, we employ classical Huffman data compression algorithm.  In short, in the first level compression is done at word level and in the second level it is at character level.  Interestingly, this two level compression technique for domain specific text outperforms classical Huffman technique.  To support our claim, we have presented both theoretical and simulation results for domain specific texts.\\
{\bf Road map:} In Section \ref{theoryresults}, we present theoretical results of our framework.  Simulation results of our proposed algorithm is presented in Section \ref{simulationresults}.  We conclude this section with a flowchart describing how our proposed approach is employed for compression and decompression stages during data transfer.
\section{Hierarchical Huffman Encoding: Theory and Simulation}
\label{theoryresults}
In this section, we present a theoretical study of Hierarchical Huffman encoding followed by simulation results.  We first discuss our approach to find frequent patterns following which we present a polynomial-time algorithm for Hierarchical Huffman encoding and it is polynomial in the size of the input text.  In the subsequent sections, we present an implementation using Hash table data structure followed by the run-time analysis of the algorithm.  Further, we show that Hierarchical Huffman encoding achieves better compression ratio than classical Huffman for Domain specific text.  We support our theoretical study by a thorough simulation of Hierarchical-Huffman algorithm for various Domain specific text inputs. Our findings, both theoretical and simulation reveal that the proposed two-level compression outperforms classical Huffman encoding algorithm.
\subsection{Identifying Frequent Patterns from the text}
We mine the input text to identify frequent patterns.  We make use of Python Natural Language Processing Toolkit (NLTK). 
\begin{itemize}
\item The NLTK toolkit supports the reading of the training data into the so called Corpus Reader, a specialized object class to enable faster access to large text files stored on the secondary memory.
\item The training data are loaded into the Corpus Reader in NLTK along with parameters {\em Length} and {\em Frequency}.  The clustering is done based on the parameter values.  The words extracted from the text are clustered into 4 clusters based on the parameters provided.
\item The 4 clusters are {\em infrequent short}, {\em infrequent long}, {\em frequent short}, and {\em frequent long} clusters. The NLTK clustering algorithm is so designed to choose only the {\em frequent long} cluster.  The other 3 clusters are avoided because their contribution in the improvement of compression ratio is minimal.  i.e., the infrequent long and infrequent short clusters are not replaced with special strings and they are simply passed to the next stage of algorithm.   Similarly, the frequent short patterns are also not replaced with special strings as the overhead for replacing short pattern with a special string is more than leaving the pattern as such in the text.
\item Since our approach is for a domain specific text, we also append the keywords (frequent patterns) from the standard text books to our corpus generated out of NLTK algorithm.  The complete set of the frequent patterns is now sorted in the order of their length and written to a file further processing.
\end{itemize}
\begin{figure}
\begin{center}
\includegraphics[scale=0.6]{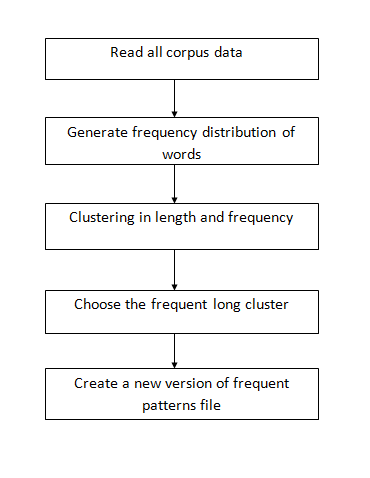}
\includegraphics[scale=0.6]{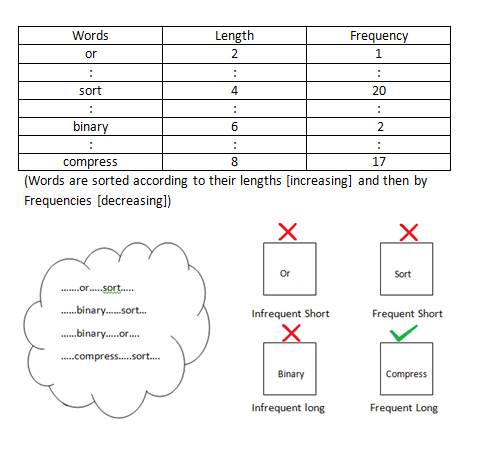}
\caption{The Flowchart of Mining Algorithm to find Clusters}
\end{center}
\end{figure}
\subsection{Hierarchical Huffman Algorithm}
Hierarchical Huffman algorithm for domain specific texts is presented in Algorithm \ref{tree-alg5}.
\begin{algorithm}[h]
\caption{Hierarchical Huffman Encoding {\tt Hierarchical-Huffman(Text T)}} \label{tree-alg5}
\begin{algorithmic}[1]
\STATE{Get Patterns of length at least 3 using Frequent Pattern Algorithm to populate {\em PATTERN-DATABASE}}
\STATE Also, get keywords of length at least 3 from standard texts (specific to domain) to populate {\em PATTERN-DATABASE}
\STATE Let $P=\{P_1,P_2,\ldots,P_k\}$ denote the set of patterns such that $|P_i| \geq 3$ and $R=\{r_1,\ldots,r_k\}$ denote the set of replacement strings for encoding.
\WHILE{text $T$ is not exhausted}
\FOR{each word $w$ of $T$, check whether $w$ is a pattern in $P$}
\IF{$w$ is $P_i$ for some $i$}
\STATE{Replace $P_i$ by $r_i$ in $T$}
\ELSE
\STATE{}
\ENDIF
\ENDFOR
\ENDWHILE
\STATE{Let $T_{level1}$ is the updated text of $T$}
\STATE{Call {\tt Classical-Huffman($T_{level1})$}}
\STATE{Let $T_H$ be the Huffman tree corresponding to $T_{level1}$ and $C_H$ be the corresponding codes for encoding}
\end{algorithmic}
\end{algorithm}
\subsection{Implementation of Hierarchical-Huffman}
\begin{itemize}
\item We maintain a Hash-table to store the set $P$ of patterns obtained.  Against, each pattern $P_i$, we store the replacement string $r_i$.  For a given $P_i$, the location in Hash table is identified using the function {\em MAP-CONTAINER()} available in UNIX systems.
\item At the decoding stage, for a given $r_i$, we can uniquely retrieve $P_i$ using a bijective function.  
\item To implement {\em classical-Huffman}, we make use of a Min-heap data structure.
\end{itemize}
\subsection{Run-time Analysis of Hierarchical-Huffman}
Let the size of the text $T$ be $n$, $n$ being the number of characters in $T$.  Clearly, $|P| \leq n$.  On an average, the Hash table operations {\em Insert} and {\em Search} take $O(1)$ time.  The operations supported by Min-heap for Classical Huffman can be implemented in $O(n\log n)$ time.  Therefore, the overall time-complexity is $O(n \log n)$.  In the worst case, the Hash table operations incur $O(n)$ which is the size of the Hash table.  The overall run-time is still  $O(n \log n)$.  The data structure {\tt B-trees} can also be used to maintain the set of patterns instead of Hash table.  For B-trees, the dictionary operations incur $O(\log n)$ time. Note that the overall run-time is still $O(n.\log n)$.
\subsection{Hierarchical Huffman Outperforms Classical Huffman}
In this section, we show that Hierarchical Huffman achieves a better compression ratio than classical Huffman for domain specific input text. 
\begin{theorem}
Hierarchical Huffman Outperforms Classical Huffman.
\end{theorem}
\begin{proof}
Consider the set $P=\{P_1,P_2,\ldots,P_k\}$ of patterns such that $|P_i| \geq 3$.  Since each $P_i$ occurs at random in text $T$, we denote the number of occurrences of each $P_i$ using a random variable.  Let $X_i$ be a random variable denoting the number of occurrences of $P_i$ in $T$.  In $T_{level1}$, each $P_i$ is replaced with a special string $r_i$ such that $|r_i| \leq |P_i|$.  This is an invariant maintained by our algorithm during each scan of the text $T$.  Moreover, the patterns of length at most two are discarded and not stored in the Hash table.  The set $P$ is organized in such  a way that for each $P_i$ and $P_j$, $i < j$, the replacement strings $r_i$ and $r_j$ are such that $|r_i| \leq |r_j|$.  Therefore, the size of the text resulting from first-level compression is 
\begin{itemize}
\item $T_{level1}= r_1.X_1+\ldots +r_k.X_k + T_{infrequent} + T_{shortpattern}$, where $T_{shortpattern}$ denotes patterns of length at most 2 and $T_{infrequent}$ denotes patterns that are infrequent. 
\item The compressed text $T_{level1}$ contains replacement strings for patterns which are frequent (say, for example, it occurs at least 10 times in $T$) and each such pattern is of length at least 3 in $T$.
\item Since $|r_i| \leq |P_i|$, $T_{level1} \leq T$, where $T$ is the original text size.  i.e., to say that the size of the text resulting from first level compression (replacement of patterns by special strings) is at most the input text.
\item In the second level compression, $T_{level1}$ is given as an input to Classical-Huffman.  It is a well-known fact that Classical Huffman is optimal with respect to prefix encoding.  Therefore, Hierarchical Huffman is optimal.  Hence, the claim follows. \qed
\end{itemize}
\end{proof}
{\bf Inference:} Note that if $|r_i| < |P_i|$, for many $P_i$'s and the frequency count of each $P_i$ is also higher, then the size of the compressed text $T_{level1} << T$. \\ \\
{\bf Compression Ratio:} {\tt Compression Ratio}= $\frac{\mbox{Performance of Hierarchical Huffman}}{\mbox{Performance of Classical Huffman}}$, where \\ \\
Performance of Hierarchical Huffman=$\frac{\mbox{Two level compression on $T$}}{\mbox{Input Text $T$}}$ \\ \\
Performance of Classical Huffman=$\frac{\mbox{Classical Huffman on $T$}}{\mbox{Input Text $T$}}$. \\ \\ 
Essentially, Compression Ratio=$\frac{\mbox{Two level compression on $T$}}{\mbox{Classical Huffman on $T$}}$\\\\
Two level compression on $T$ refers to the replacement of $P_i$'s by $r_i$'s followed by Classical Huffman.  Clearly, the lower the ratio, the better is the two-level compression.  
\section{Simulation of Hierarchical-Huffman Algorithm}
\label{simulationresults}
In this section, we validate our theoretical study by a thorough simulation of Hierarchical Huffman for different input texts. We have considered input texts from Computer Science Domain.  Keywords from standard computer science texts are also considered as patterns for the study.  Our simulation includes various text files whose size ranging from 500kB to 2MB.  Various plots and findings from simulation are given below:
\begin{itemize}
\item  Figure \ref{fig1}, illustrates the performance of Hierarchical Huffman and Classical Huffman for different text inputs.  From the plot we infer that for large input texts, the performance of Hierarchical Huffman is much better than the performance of classical Huffman.  This is true because, for large input texts, frequency of patterns increases which lead to reduction in the size of the text output by two-level compression routine.\\
\begin{figure}
\begin{center}
\caption{Input Text vs Performance} 
\includegraphics[scale=0.6]{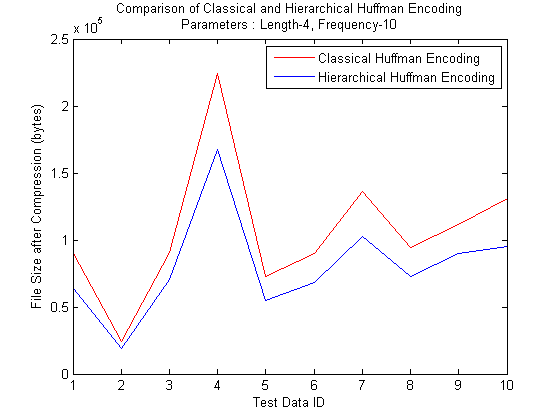}\label{fig1}
\end{center}
\end{figure}
\item In Figure \ref{fig2}, we show the plot between compression ratio and various input texts.  Recall that, from our theoretical study we infered that the higher the input text size, the better the compression ratio.  This is precisely evident in Figure \ref{fig2} as well.
\begin{figure}
\begin{center}
\caption{Input Text vs Compression Ratio} 
\includegraphics[scale=0.6]{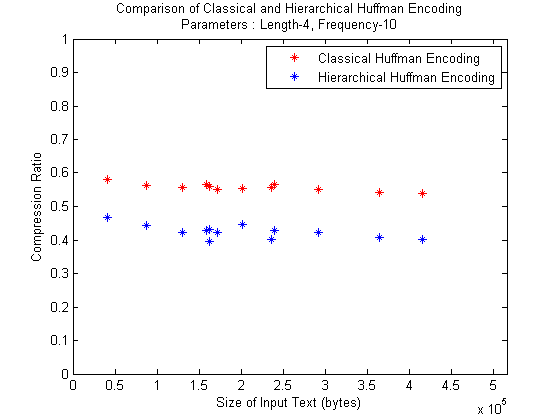}\label{fig2}
\end{center}
\end{figure}
\item Theoretical analysis reveals that to get a good compression ratio, the input text must contain smaller patterns with large frequencies.  For input texts with large patterns, it appears like Hierarchical Huffman as good as classical Huffman and this observation is clearly evident from our simulation study and it is illustrated in Figure \ref{fig3}
\begin{figure}
\begin{center}
\caption{Comparison of Two Techniques for Large Patterns} \label{fig3}
\includegraphics[scale=0.6]{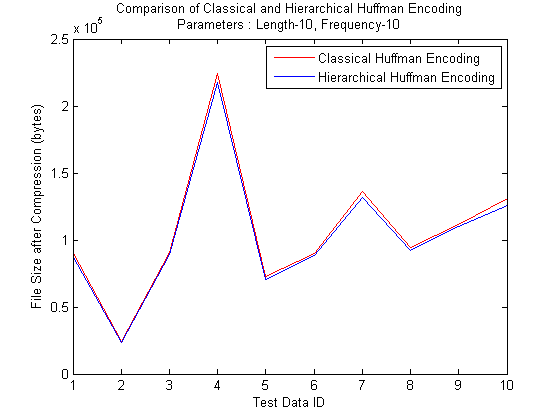}
\end{center}
\end{figure}
\item Note that during data transmission, the Hash table of reference strings of patterns will be sent along with the input text.  Since it is a domain-specific downloads, it is sufficient to send the Hash table exactly once.  Although, it is an overhead, the performance can be seen if the number of downloads is large.  This is precisely illustrated in Figure \ref{fig4}.  Due to this overhead, if the number of downloads is small, then classical Huffman performs better than the Hierarchical Huffman.  There is {\em critical point} which denotes the minimum number of downloads beyond which Hierarchical Huffman outperforms classical Huffman.
\begin{figure}
\begin{center}
\caption{Identifying Critical Point for the two Compression Schemes} \label{fig4}
\includegraphics[scale=0.6]{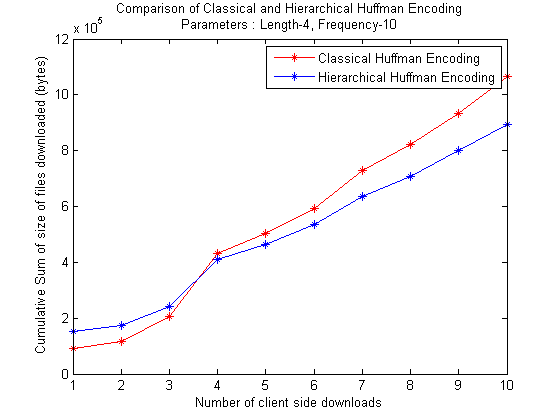}
\end{center}
\end{figure}
\item As mentioned before, shorter patterns with good frequency count  gives good compression ratio than larger patterns with same frequency count.  This observation is validated for different texts and it is illustrated in Figure \ref{fig5}.  The plot is done by taking average of the compression ratios for different input texts versus the vector (length of the pattern,frequency of the pattern).
\begin{figure}
\begin{center}
\caption{Effects of Shorter vs Larger Patterns on the Compression Ratio} \label{fig5}
\includegraphics[scale=0.6]{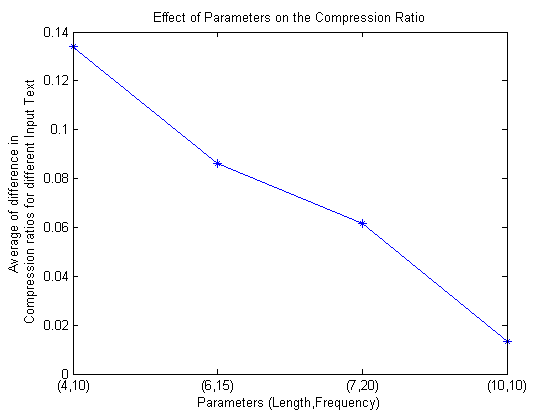}
\end{center}
\end{figure}
\end{itemize}
\subsection{Conclusions and Further Research}
In this paper, we have proposed a framework for data compression for domain specific input text.  Our framework involves two-level compression in which the first level is done at word level and the second level is at character level.   We have also shown that this two-level approach outperforms classical Huffman for domain specific input text.   All our claims are supported by theoretical results and validated with a thorough simulation.   An interesting problem for further research is to analyse our two-level scheme for video compression.

\bibliographystyle{splncs}
\bibliography{illam-ref}   

\begin{thebibliography}{4}

\bibitem{dm}
Naren Ramakrishnan, Ananth Grama: Data Mining: From Serendipity to Science, IEEE Computers 32(8): 34-37 (1999).
\bibitem{datamining}
Jiawei Han, Micheline Kamber: Data Mining: Concepts and Techniques, Morgan Kaufmann 2000, ISBN 1-55860-489-8.
\bibitem{aho}
Aho and Hopcroft: Data Structures and Algorithms, Academic Press, 1990.
\bibitem{genetic}
David E. Goldberg: Genetic Algorithms in Search Optimization and Machine Learning. Addison-Wesley 1989, ISBN 0-201-15767-5
\bibitem{amar}
Robert Franceschini, Amar Mukherjee: Data Compression Using Encrypted Text. ADL 1996: 130-138
\bibitem{lossless}
Ian H. Witten, Zane Bray, Malika Mahoui, W. J. Teahan: Text Mining: A New Frontier for Lossless Compression. Data Compression Conference 1999: 198-207
\bibitem{pitman}
P. Nagabhushan, S. Murali: Pitman Shorthand Inspired Model for Plain Text Compression. ICDAR 2001: 132
\bibitem{tara}
TARA: An Algorithm for Fast Searching of Multiple Patterns on Text Files, Technical Report, Turkish Army Gendarme Headquarter,Bestepe,Ankara,Turkey, 2007
\bibitem{corpus}
Weifeng Sun, Nan Zhang, Amar Mukherjee: A Dictionary-Based Multi-Corpora Text Compression System. DCC 2003: 448
\bibitem{bang}
Md. Ziaul Karim Zia, Dewan Md. Fayzur Rahman, and Chowdhury Mofizur Rahman, Two-Level Dictionary Based Text Compression Scheme, Proceedings of 11th International Conference on Computer and Information Technology (ICCIT 2008), 2008, pp. 13-18
\end{thebibliography}
\end{document}